\documentclass[10pt,usletter]{article}
\usepackage[totalwidth=450pt,totalheight=640pt]{geometry}
\date{}
\usepackage{amsthm}

\usepackage{url}

\usepackage{graphicx,color}

\usepackage{booktabs}
\usepackage{amsmath}
\usepackage{bm}
\usepackage{tikz}

\newenvironment{myquote}{\begin{center}
    \begin{minipage}{.80\linewidth}}{\end{minipage}\end{center}}

\DeclareRobustCommand{\exqed}{%
  \ifmmode
    \eqno \def\@badmath{$$}
    \let\eqno\relax \let\leqno\relax \let\veqno\relax
    \hbox{\ensuremath{\dashv}}%
  \else
    \leavevmode\unskip\penalty9999 \hbox{}\nobreak\hfill
    \quad\hbox{\ensuremath{\dashv}}%
  \fi
}
\makeatother

\newtheorem{lemma}{Lemma}
\newtheorem{proposition}{Proposition}
\theoremstyle{remark}
\newtheorem{example}{Example}

\renewenvironment{example}{\begin{ex}}{\hfill
    $\dashv$\end{ex}\medskip}
\newenvironment{examplenoqed}{\begin{ex}}{\end{ex}}

\newcommand{\cwd}{\text{\normalfont cwd}}
\renewcommand{\P}{\text{\normalfont P}}
\newcommand{\NP}{\text{\normalfont NP}}

\newcommand{\DDD}{\mathcal{D}}

\newcommand{\SB}{\{\,}%
\newcommand{\SM}{\;{:}\;}%
\newcommand{\SE}{\,\}}%
\newcommand{\Card}[1]{|#1|}

\newcommand{\hy}{\hbox{-}\nobreak\hskip0pt}

\newcommand{\Fder}{F_{\text{\normalfont der}}}

\newcommand{\comps}{\text{cmp}}
\newcommand{\groups}{\text{grp}}

\begin{document}

\title{A SAT Approach to Clique-Width}
\author{Marijn J. H. Heule\thanks{Research supported in part by the National Science
   Foundation under grant CNS-0910913 and DARPA
    contract number N66001-10-2-4087.}\\
\small  Department of Computer Sciences\\[-3pt]
\small  The University of Texas at Austin, USA
  \and Stefan Szeider\thanks{Research supported by the ERC, grant
    reference 239962 (COMPLEX REASON).}\\
\small  Institute of Information Systems\\[-3pt]
\small  Vienna University of Technology,
  Vienna, Austria 
}

\maketitle

\thispagestyle{empty}

\begin{abstract}
  Clique-width is a graph invariant that has been widely studied in
  combinatorics and computer science.  However, computing the
  clique-width of a graph is an intricate problem, the exact
  clique-width is not known even for very small graphs.  We present a
  new method for computing the clique-width of graphs based on an
  encoding to propositional satisfiability (SAT) which is then
  evaluated by a SAT solver. Our encoding is based on a reformulation
  of clique-width in terms of partitions that utilizes an efficient
  encoding of cardinality constraints.  Our SAT-based method is the
  first to discover the exact clique-width of various small graphs,
  including famous graphs from the literature as well as random graphs
  of various density. With our method we determined the smallest graphs
  that require a small pre-described clique-width.
\end{abstract}

\section{Introduction}\label{sect:intro}

Clique-width is a fundamental graph invariant that has been widely
studied in combinatorics and computer science. Clique-width measures
in a certain sense the ``complexity'' of a graph. It is defined via a
graph construction process involving four operations where only a
limited number of vertex labels are available; vertices that share the
same label at a certain point of the construction process must be
treated uniformly in subsequent steps.  This graph composition
mechanism was first considered by Courcelle, Engelfriet, and Rozenberg
\cite{CourcelleEngelfrietRozenberg90,CourcelleEngelfrietRozenberg93}
and has since then been an important topic in combinatorics and
computer science.

Graphs of small clique-width have  advantageous algorithmic properties. 
Algorithmic meta-theorems show that large classes of
$\NP$-hard optimization problems and \#P-hard counting problems can be
solved in \emph{linear time} on classes of graphs of bounded
clique-width~\cite{CourcelleMakowskyRotics00,CourcelleMakowskyRotics01}. Similar
results hold for the graph invariant \emph{treewidth}, however,
clique-width is more general in the sense that graphs of small
treewidth also have small clique-width, but there are graphs of small
clique-width but arbitrarily high treewidth
\cite{CourcelleOlariu00,CornelRotics05}. Unlike treewidth, dense
graphs (e.g., cliques) can also have small clique-width.

All these algorithms for graphs of small clique-width require that a
certificate for the graph having small clique-width is provided.
However, it seems that computing the certificate, or just deciding
whether the clique-width of a graph is bounded by a given number, is a
very intricate combinatorial problem.  More precisely, given a graph
$G$ and an integer $k$, deciding whether the clique-width of $G$ is at
most~$k$ is $\NP$\hy complete~\cite{FellowsRosamondRoticsSzeider09}.
Even worse, the clique-width of a graph with $n$ vertices of degree
greater than~2 cannot be approximated by a polynomial-time algorithm
with an absolute error guarantee of $n^\epsilon$ unless $\P=\NP$,
where $0\leq \epsilon < 1$~\cite{FellowsRosamondRoticsSzeider09}.  In
fact, it is even unknown whether graphs of clique-width at most~$4$
can be recognized in polynomial time~\cite{CornelEtal12}.  There are
approximation algorithms with an exponential error that, for
fixed~$k$, compute $f(k)$\hy expressions for graphs of clique-width at
most~$k$ in polynomial time (where $f(k)=(2^{3k+2}-1)$ by~\cite{OumSeymour06},
and $f(k)=8^k-1$ by~\cite{Oum08}).

Because of this intricacy of this graph invariant, the exact
clique-width is not known even for very small graphs.


\paragraph{Clique-width via SAT.}

We present a new method for determining the
clique-width based on a sophisticated SAT encoding which entails the
following  ideas:
\begin{enumerate}
\item \emph{Reformulation}. The conventional construction method for
  determining the clique-width of a graph consists of many steps. In
  the worst case, the number of steps is quadratic in the number of
  vertices.  Translating this construction method into SAT would
  result in large instances, even for small graphs. We reformulated
  the problem in such a way that the number of steps is less than the
  number of vertices.  The alternative construction method allows us
  to compute the clique-width of much larger graphs.

\item \emph{Representative encoding}. Applying the frequently-used
  direct encoding~\cite{Walsh00} on the reformulation results in
  instances that have no arc consistency~\cite{Gent02}, i.e., unit
  propagation may find conflicts much later than required. We
  developed the representative encoding that is compact and realizes
  arc consistency.

\end{enumerate}
 
\paragraph{Experimental Results.}

The implementation of our method allows us for the first time to
determine the exact clique-width of various graphs, including famous
graphs known from the literature, as well as random graphs of various
density.

\begin{enumerate}
\item \emph{Clique-width of small Random Graphs}.  We determined
  experimentally how the clique-width of random graphs depends on the
  density. The clique-width is small for dense and sparse graphs and
  reaches its maximum for edge-probability $0.5$. The larger $n$, the
  steeper the increase towards $0.5$.  These results complement the
  asymptotic results of Lee et al.~\cite{LeeLeeOum12}.

\item \emph{Smallest Graphs of Certain Clique-width}. In general it is
  not known how many vertices are required to form a graph of a certain
  clique-width. We provide these numbers for clique-width $k\in
  \{1,\dots,7\}$. In fact, we could compute the total number of
  connected graphs (modulo isomorphism) with a certain clique-width
  with up to 10 vertices. For instance, there are only~7 connected
  graphs  with 8 vertices and clique-width 5 (modulo isomorphism),
  and no graphs with 9 vertices and clique-width 6. There are 68 graphs
  with 10 vertices and clique-width 6. The smallest one has 18 edges. 
  
\item \emph{Clique-width of Famous Named Graphs}. Over the last 50
  years, researchers in graph theory have considered a large number of
  special graphs. These special graphs have been used as
  counterexamples for conjectures or for showing the tightness of
  combinatorial results.  We considered several prominent graphs from
  the literature and computed their exact clique-width. These results
  may be of interest for people working in combinatorics and graph
  theory.
\end{enumerate}

\paragraph{Related Work.} We are not aware of any implemented
algorithms that compute the clique-width exactly or
heuristically. However, algorithms have been implemented that compute
upper bounds on other width-based graph invariants, including
\emph{treewidth}~\cite{DowKorf07,GogateDechter04,KosterBodlaenderHoesel01},
\emph{branchwidth}~\cite{SmithUlusalHicks12},
\emph{Boolean-width}~\cite{HvidevoldEtal11}, and
\emph{rank-width}~\cite{Beyss2013}.  Samer and Veith
\cite{SamerVeith09} proposed a SAT encoding for the exact computation
of treewidth.  Boolean-width and rank-width can be used to approximate
clique-width, however, the error can be exponential in the
clique-width; in contrast, treewidth and branchwidth can be
arbitrarily far from the clique-width, hence the approximation error
is unbounded~\cite{BuixuanTelleVatshelle11}.

Our SAT encoding is based on a new characterization of clique-width
that is based on partitions instead of labels. A similar
partition-based characterization of clique-width, has been proposed by
Heggernes et al.~\cite{HeggernesMeisterRotics11}.  There are two main
differences to our reformulation. Firstly, our characterization of
clique-width uses three individual properties that can be easily
expressed by clauses.  Secondly, our characterization admits the
``parallel'' processing of several parts of the graph that are later
joined together.

\section{Preliminaries}
\label{sec:prelim}

\subsection{Formulas and Satisfiability}
We consider propositional formulas in Conjunctive Normal Form
(\emph{CNF formulas}, for short), which are conjunctions of clauses,
where a clause is a disjunction of literals, and a literal is a
propositional variable or a negated propositional variables.  A CNF
formula is \emph{satisfiable} if its variables can be assigned true or
false, such that each clause contains either a variable set to true or
a negated variable set to false.  The satisfiability problem (SAT)
asks whether a given formula is satisfiable.


\subsection{Graphs and Clique-width}

All graphs considered are finite, undirected, and without self-loops.
We denote a graph $G$ by an ordered pair $(V(G),E(G))$ of its set of
vertices and its set of edges, respectively. An edge between vertices
$u$ and $v$ is denoted $uv$ or equivalently $vu$. For basic
terminology on graphs we refer to a standard text book~\cite{Diestel00}.

Let $k$ be a positive integer. A $k$-\emph{graph} is a graph whose
vertices are labeled by integers from $\{1,\dots,k\}$.  We consider an
arbitrary graph as a $k$\hy graph with all vertices labeled by~$1$.
We call the $k$-graph consisting of exactly one vertex $v$ (say,
labeled by $i$) an \emph{initial} $k$-graph and denote it by $i(v)$.
The \emph{clique-width} of a graph $G$ is the smallest integer $k$
such that $G$ can be constructed from initial $k$-graphs by means of
repeated application of the following three operations.
\begin{enumerate}
\item Disjoint union (denoted by $\oplus$);
\item Relabeling: changing all labels $i$ to $j$ (denoted by $\rho_{i
    \rightarrow j}$);
\item Edge insertion: connecting all vertices labeled by $i$ with all
  vertices labeled by $j, i \neq j$ (denoted by $\eta_{i,j}$ or
  $\eta_{j,i}$); already existing edges are not doubled.
\end{enumerate}
A construction of a $k$\hy graph using the above operations can be
represented by an algebraic term composed of $\oplus$, $\rho_{i
  \rightarrow j}$, and $\eta_{i,j}$ ($i,j\in \{1,\dots,k\}$, and
$i\neq j$).  Such a term is called a \emph{$k$\hy expression}
defining~$G$.  Thus, the clique-width of a graph $G$ is the smallest
integer $k$ such that $G$ can be defined by a $k$\hy expression.
\begin{example}\label{ex:p4}
  The graph $P_4=(\{a,b,c,d\},\{ab, bc, cd\})$ is defined by the
  $3$\hy expression%
 \[ \eta_{2,3}( \rho_{2\rightarrow 1}(
    \eta_{2,3}( \eta_{1,2}(1(a) \oplus 2(b)) \oplus 3(c) ) ) \oplus
    2(d) ).
\]
Hence $\cwd(P_4)\leq 3$. In fact, one can show that $P_4$ it has no
$2$\hy expression, and thus $\cwd(P_4)=3$~\cite{CourcelleOlariu00}.
\end{example}

\subsection{Partitions}

As partitions play an important role in our reformulation of
clique-width, we recall some basic terminology. A \emph{partition} of
a set $S$ is a set $P$ of nonempty subsets of $S$ such that any two
sets in $P$ are disjoint and $S$ is the union of all sets in $P$. The
elements of $P$ are called \emph{equivalence classes}.  Let $P,P'$ be
partitions of $S$. Then $P'$ is a \emph{refinement} of $P$ if for any
two elements $x,y\in S$ that are in the same equivalence class of $P'$
are also in the same equivalence class of $P$ (this entails the case
$P=P'$).
 
\section{A Reformulation of Clique-width without Labels}

Initially, we developed a SAT encoding of clique-width based on
$k$-expressions.  Even after several optimization steps, this encoding
was only able to determine the clique-width of graphs consisting of at
most 8 vertices.  We therefore developed a new encoding based on a 
reformulation of clique-width which does not use $k$\hy
expressions. In this section we explain this reformulation, in the
next section we will discuss how it can be encoded
into~SAT efficiently.


Consider a finite set $V$, the \emph{universe}.  A \emph{template} $T$
consists of two partitions $\comps(T)$ and $\groups(T)$ of~$V$. We
call the equivalence classes in $\comps(T)$ the \emph{components} of
$T$ and the equivalence classes in $\groups(T)$ the \emph{groups} of
$T$.  For some intuition about these concepts, imagine that  components represent
induced subgraphs and that groups represent sets of vertices in some component 
with the same label in a $k$-expression. A \emph{derivation} of length $t$ is a finite sequence
$\DDD=(T_0,\dots,T_t)$ satisfying the following conditions.
\begin{enumerate}
\item[D1] ~~~$\Card{\comps(T_0)}=\Card{V}$ and $\Card{\comps(T_t)}=1$.
\item[D2] ~~~$\groups(T_{i})$ is a refinement of $\comps(T_{i})$, $0\leq
  i \leq t$.
\item[D3] ~~~$\comps(T_{i-1})$ is a refinement of  $\comps(T_{i})$, $1\leq i \leq
  t$.
\item[D4] ~~~$\groups(T_{i-1})$ is a refinement of $\groups(T_{i})$, $1\leq i
  \leq t$.
\end{enumerate}
We would like to note that D1 and D2 together imply that
$\Card{\groups(T_0)}=\Card{V}$. Thus, in the first template $T_0$ all
equivalence classes (groups and components) are singletons, and when
we progress through the derivation, some of these sets are merged,
until all components are merged into a single component in the last
template $T_t$.

The \emph{width} of a component $C\in \comps(T)$ is the number of
groups $g\in \groups(T)$ such that \mbox{$g\subseteq C$}.  The width
of a template is the maximum width over its components, and the width
of a derivation is the maximum width over its templates.  A
\emph{$k$\hy derivation} is a derivation of width at most~$k$.
A derivation $\DDD=(T_0,\dots,T_t)$ is a derivation \emph{of} a
graph $G=(V,E)$ if $V$ is the universe of the derivation and the
following three conditions hold for all $1\leq i \leq t$.
\begin{description}
\item{\emph{Edge Property}:} For any two vertices $u,v\in V$ such that
  $uv\in E$, if $u,v$ are in the same group in $T_i$, then $u,v$
  are in the same component in $T_{i-1}$.
\item{\emph{Neighborhood Property}:} For any three vertices $u,v,w\in V$
  such that $uv\in E$ and $uw\notin E$, if $v,w$ are in the same group
  in $T_i$, then $u,v$ are in the same component in $T_{i-1}$.
\item{\emph{Path Property}:} For any four vertices $u,v,w,x\in V$,
  such that $uv, uw, vx \in E$ and $wx \notin E$, if $u,x$ are in the
  same group in $T_i$ and $v,w$ are in the same group in $T_i$, then
  $u,v$ are in the same component in $T_{i-1}$.
\end{description}
The neighborhood property and the path
property could be merged into a single property if we do not insist that
all mentioned vertices are distinct. However, two separate properties
provide a more compact SAT encoding.

The following example illustrates that a derivation can define more
than one graph, in contrast to a $k$\hy expression, which defines
exactly one graph.

\begin{example}\label{ex:derivation}
  Consider the derivation $\DDD=(T_0,\dots,T_3)$ with
  universe $V=\{a,b,c,d\}$ and
  \begin{myquote}
  $\begin{array}[b]{lclclcl}
    \comps(T_0)&=&\{\{a\},\{b\},\{c\},\{d\}\}, 
    & \quad & 
    \groups(T_0)&=& \{\{a\},\{b\},\{c\},\{d\}\}, \\  
    \comps(T_1)&=&\{\{a,b\},\{c\},\{d\}\}, 
    & \quad & 
    \groups(T_1)&=&\{\{a\},\{b\},\{c\},\{d\}\},\\ 
    \comps(T_2)&=&\{\{a,b,c\},\{d\}\}, 
    & \quad & 
    \groups(T_2)&=& \{\{a\},\{b\},\{c\},\{d\}\}, \\  
    \comps(T_3)&=&\{\{a,b,c,d\}\}, 
    & \quad & 
    \groups(T_3)&=& \{\{a,b\},\{c\},\{d\}\}.  
  \end{array}$ 
\end{myquote}
The width of $\DDD$ is $3$.  Consider the graph
$G=(V,\{ab,ad,bc,bd\})$. To see that~$\DDD$ is a $3$\hy derivation
of~$G$, we need to check the edge, neighborhood, and path
properties. We observe that $a,b$ are the only two vertices such that
$ab\in E(G)$ and both vertices appear in the same group of some $T_i$
(here, we have $i=3$).  
To check the edge property, we
only need to verify that $a,b$ are in the same component of~$T_2$,
which is true. For the neighborhood property, the only
relevant choice of three vertices is $a,b,c$ ($bc\in E(G)$, $ac\notin
E(G)$, and $a,b$ in a group of~$T_3$). The neighborhood
property requires that $b,c$ are in the same component in $T_2$, which
is the case.  The path property is satisfied since there is no
template in which two pairs of vertices belong to the same group,
respectively.

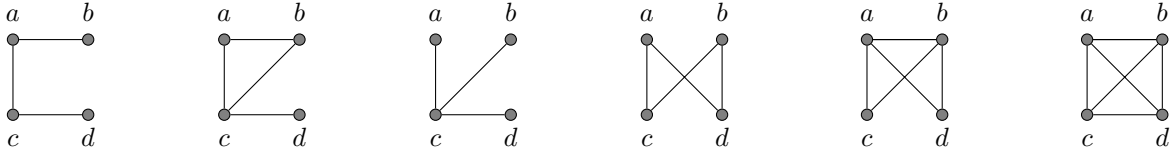
\begin{figure}[ht]
\centering
 \tikzstyle{every circle node}=[circle,draw,inner sep=1.5pt, fill=gray]

\begin{tikzpicture}
 \draw 
 (0,1)      node[circle] (a) {} node[above] {$\strut a$} 
 (1,1)      node[circle] (b) {} node[above] {$\strut b$} 
 (0,0)      node[circle] (c) {} node[below] {$\strut c$} 
 (1,0)      node[circle] (d) {} node[below] {$\strut d$} 
 (b)--(a)--(c)--(d)
;
\end{tikzpicture}\hfill
\begin{tikzpicture}
 \draw 
 (0,1)      node[circle] (a) {} node[above] {$\strut a$} 
 (1,1)      node[circle] (b) {} node[above] {$\strut b$} 
 (0,0)      node[circle] (c) {} node[below] {$\strut c$} 
 (1,0)      node[circle] (d) {} node[below] {$\strut d$} 
 (a)--(b)--(c)--(a) (c)--(d)
;
\end{tikzpicture}\hfill
\begin{tikzpicture}
 \draw 
 (0,1)      node[circle] (a) {} node[above] {$\strut a$} 
 (1,1)      node[circle] (b) {} node[above] {$\strut b$} 
 (0,0)      node[circle] (c) {} node[below] {$\strut c$} 
 (1,0)      node[circle] (d) {} node[below] {$\strut d$} 
 (a)--(c)--(d) (b)--(c)
;
\end{tikzpicture}\hfill
\begin{tikzpicture}
 \draw 
 (0,1)      node[circle] (a) {} node[above] {$\strut a$} 
 (1,1)      node[circle] (b) {} node[above] {$\strut b$} 
 (0,0)      node[circle] (c) {} node[below] {$\strut c$} 
 (1,0)      node[circle] (d) {} node[below] {$\strut d$} 
 (a)--(c)--(b)--(d)--(a)
;
\end{tikzpicture}
\hfill
\begin{tikzpicture}
 \draw 
 (0,1)      node[circle] (a) {} node[above] {$\strut a$} 
 (1,1)      node[circle] (b) {} node[above] {$\strut b$} 
 (0,0)      node[circle] (c) {} node[below] {$\strut c$} 
 (1,0)      node[circle] (d) {} node[below] {$\strut d$} 
 (a)--(b)--(d)--(a)--(c)--(b)--(a)
;
\end{tikzpicture}
\hfill
\begin{tikzpicture}
 \draw 
 (0,1)      node[circle] (a) {} node[above] {$\strut a$} 
 (1,1)      node[circle] (b) {} node[above] {$\strut b$} 
 (0,0)      node[circle] (c) {} node[below] {$\strut c$} 
 (1,0)      node[circle] (d) {} node[below] {$\strut d$} 
 (a)--(b)--(d)--(a)--(c)--(b)--(a) (c)--(d)
;
\end{tikzpicture}
  \vspace{5pt}
  \caption{All connected graphs with four vertices (up to isomorphism). 
  The 3-derivation of Example~\ref{ex:derivation} defines all six graphs. The
  clique-width of all but the first graph is 2.}  
  \label{fig:four}
\end{figure}

Similarly we can verify that $\DDD$ is a derivation of the graph
$G'=(V,\{ab$, $bc$, $cd\})$.  In fact, for all connected graphs with four
vertices, there exists an isomorphic graph that is defined by $\DDD$
(see Figure~\ref{fig:four}).  However, $\DDD$ is not a derivation of the
graph $G''=(V,\{ab,ac,bd,cd\})$ since the neighborhood property is
violated: $bd \in E(G'')$ and $ad \notin E(G'')$, $a,b$ belong to the
same group in $T_3$, while $a,d$ do not belong to the same component
in~$T_2$.
\end{example}


\sloppypar We call a derivation $(T_0,\dots,T_t)$ to be \emph{strict}
if $\Card{\comps(T_{i-1})} > \Card{\comps(T_{i})}$ holds for all
$1\leq i \leq t$. 
\begin{lemma}\label{lem:make-strict}
  If $G$ has a $k$\hy derivation, it has a strict $k$\hy derivation.
\end{lemma}
\begin{proof}
  Let $\DDD=(T_0,\dots,T_t)$ be a $k$\hy derivation of~$G$.  Assume
  there is some $1 \leq i\leq t$ such that
  $\comps(T_{i-1})=\comps(T_{i})$.  If also
  $\groups(T_{i-1})=\groups(T_{i})$, then $T_{i-1}=T_i$, and we can
  safely remove $T_{i-1}$ and still have a $k$\hy derivation
  of~$G$. Hence assume $\groups(T_{i-1})\neq \groups(T_{i})$.  This
  implies that $i>1$.  If $i=t$, then we can safely remove $T_t$ from
  the derivation and $(T_0,\dots,T_{t-1})$ is clearly a $k$\hy
  derivation of $G$. Hence it remains to consider the case $1 < i
  \leq t-1$.  We show that by dropping $T_i$ we get a sequence
  $\DDD'=(T_0,\dots,T_{i-1},T_{i+1},\dots,T_t)$ that is a $k$\hy
  derivation of $G$. 

  The new sequence $\DDD'$ is clearly a $k$\hy derivation.  It remains
  to verify that $\DDD'$ is a derivation of~$G$. The template
  $T_{i+1}$ is the only one where these properties might have been
  violated by the removal of $T_i$. However, since all three
  properties impose a restriction on the set of components of the
  template preceding $T_{i+1}$, and since
  $\comps(T_{i-1})=\comps(T_{i})$, the properties are not effected by
  the deletion of $T_i$. Hence $\DDD'$ is indeed a $k$\hy derivation
  of $G$.

  By repeated application of the above shortening we can turn any
  $k$\hy derivation into a strict $k$\hy derivation.
\end{proof}

\begin{lemma}\label{lem:strict-short}
  Every strict $k$\hy derivation of a graph with $n$ vertices has
  length at most $n-1$.
\end{lemma}
\begin{proof}
  Let $(T_0,\dots,T_t)$ be a strict $k$\hy derivation of a graph with
  $n$ vertices.  Since $\Card{\comps(T_{0})}=n$ and
  $\Card{\comps(T_{0})}=1$, it follows that $t\leq n-1$.
\end{proof}

\medskip

In the proofs of the next two lemmas we need the following concept of
a \emph{$k$\hy expression tree}, which is the parse tree of a $k$\hy
expression equipped with some additional information.  Let $\phi$ be a
$k$\hy expression for a graph $G=(V,E)$.  Let $Q$ be the parse tree of
$\phi$ with root~$r$.  That is, $Q$ contains a node for each occurrence
of an operation $\oplus$, $\rho_{i \rightarrow j}$, and $\eta_{i,j}$
in $\phi$ and for each initial $k$\hy graph $i(v)$ in $\phi$; the
initial $k$\hy graphs are the leaves of $Q$, and the other nodes have
as children the nodes which represent the two subexpressions of the
respective operation.  Consider a node $q$ of $Q$ and let $\phi_q$ be
the subexpression of $\phi$ whose parse tree is the subtree of $Q$
rooted at $q$. Then $q$ is labeled with the $k$\hy graph $G_q$
constructed by the $k$\hy expression~$\phi_q$.  Thus the leaves of $Q$
are labeled with initial $k$\hy graphs and the root $r$ is labeled
with a labeled version of~$G$.  We call a non-leaf node of $Q$ an
$\oplus$\hy node, $\eta$\hy node, or $\rho$\hy node, according to the
operation it represents.

One $\oplus$\hy node of the parse tree can represent several directly
subsequent $\oplus$\hy operations (e.g., the operation $(x \oplus y)
\oplus z$ can be represented by a single node with three children).
For technical reasons we will also allow $\oplus$\hy nodes with a
single child.

Each $k$\hy expression gives rise to a $k$\hy expression tree where
each $\oplus$\hy node has no $\oplus$\hy nodes as children, let us
call such a $k$\hy expression tree to be \emph{succinct}.  Evidently,
$k$\hy expressions and their (succinct) $k$\hy expression trees can be
effectively transformed into each other.

\begin{lemma}\label{lem:expression->derivation}
  From a $k$\hy expression of a graph $G$ we can obtain a $k$\hy
  derivation of~$G$ in polynomial time.
\end{lemma}
\begin{proof}
  Let $\phi$ be a $k$\hy expression of $G=(V,E)$ and let $Q$ be the
  corresponding succinct $k$\hy expression tree with root $r$.  For a
  node $q \in V(Q)$ let $R(q)$ denote the number of $\oplus$-nodes
  that appear on the path from $r$ to $q$; thus $R(r)=1$.  We write
  $U$ and $L$ for the set of $\oplus$\hy nodes and the set of leaves
  of $Q$, respectively.  We let $t:=\max_{q\in L} R(q)$.  For $0\leq i
  \leq t$ we define $U_i=\SB q\in U \SM R(q)=t-i+1 \SE$ and $L_i=\SB
  q\in L \SM R(q)< t-i+1 \SE$. We observe that for each $v\in V$ and
  $1\leq i \leq t$ there is exactly one $q\in U_i\cup L_i$ such that
  $v\in G_q$

  We define a derivation $\DDD=(T_0,\dots,T_t)$ as follows.   For
  $0 \leq i\leq t$ we
  put $\comps(T_i)=\SB V(G_q) \SM q\in U_i \cup L_i\SE$ and
  $\groups(T_i)=\bigcup_{q\in U_i \cup L_i} \groups(G_q)$ where
  $\groups(G_q)$ denotes the partition of $V(G_q)$ into sets of
  vertices that have the same label.  By construction, $\DDD$ is a
  derivation with universe $V$. Furthermore, since $\phi$ is a $k$\hy
  expression, $\Card{\groups(G_q)}\leq k$ for all nodes $q$ of
  $Q$. Hence $\DDD$ is a $k$\hy derivation.
  It remains to show that~$\DDD$ is a $k$\hy derivation
    of~$G$. Let $1\leq i \leq t$.
 
    To show that the \emph{edge property} holds, consider two vertices
    $u,v\in V$ such that $uv\in E$ and $u,v$ are in the same group
    in~$T_i$.  Assume to the contrary that $u,v$ belong to different
    components $c_1,c_2$ in~$T_{i-1}$.  Since $u,v$ are in the same
    group in~$T_i$, they are also in the same component of $T_i$.
    Hence there is an $\oplus$\hy node $q\in U_i$ with $u,v\in
    V(G_q)\in \comps(T_i)$. Let $q_1,q_2$ be the children of $q$ with
    $V(G_{q_1})=c_1$ and $V(G_{q_2})=c_2$. Hence $uv\notin
    E(G_{q_1})\cup E(G_{q_1})$. However, since $u,v$ are in the same
    group in $T_i$, this means that $u,v$ have the same label in
    $G_q$. Thus the edge $uv$ cannot be introduced by an $\eta$\hy
    operation, and so $uv\notin E(G_{r})=E$, a contradiction. Hence
    the edge property holds.

    To show that the \emph{neighborhood property} holds, consider
    three vertices $u,v,w\in V$ such that $uv\in E$, $uw\notin E$, and
    $v,w$ are in the same group of~$T_i$. Assume to the contrary that
    $u,v$ are in different components of~$T_{i-1}$, say in components
    $c_1$ and $c_2$, respectively.  Since $v,w$ are in the same group
    of~$T_i$, they are also in the same component $c$ of $T_i$. Let
    $q\in U_i$ be the $\oplus$\hy node such that $v,w\in V(G_q)=c\in
    \comps(T_i)$, and let $q_1,q_2$ be the children of $q$ with
    $V(G_{q_1})=c_1$ and $V(G_{q_2})=c_2$.  Clearly $uv\notin
    E(G_{q_1})\cap E(G_{q_2})$, hence there must be an $\eta$\hy node
    $p$ somewhere on the path between $q$ and $r$ where the edge $uv$
    is introduced. However, since~$v$ and $w$ share the same label in
    $G_q$, they share the same label in $G_p$. Consequently, the
    $\eta$\hy operation that introduces the edge $uv$ also introduces
    the edge $uw$. However, this contradicts the assumption that
    $uw\notin E$. Hence the neighborhood property holds as well.

    To show that the \emph{path property} holds, we proceed
    similarly. Consider four vertices $u,v,w,x\in V$, such that $uv,
    uw, vx \in E$ and $wx \notin E$. Assume that $u,x$ are in the same
    group in $T_i$ and $v,w$ are in the same group in $T_i$. Assume to
    the contrary that $u,v$ are in different components of~$T_{i-1}$,
    say in components $c_1$ and $c_2$, respectively. Above we have
    shown that the neighborhood property holds. Hence we conclude that
    $u,w$ belong to the same component of $T_{i-1}$, and $v,x$ belong
    to the same component of $T_{i-1}$.  Since $u,x$ are in the same
    group in $T_i$, they are also in the same component of~$T_i$, say
    in component~$c$.  Since $u,w$ belong to the same component of
    $T_{i-1}$, they also belong to the same component of $T_i$, thus
    $w\in c$. By a similar argument we conclude that $v\in c$.  Thus
    all four vertices $u,v,w,x$ belong to~$c$.  Let $q\in U_i$ be the
    $\oplus$\hy node with $V(G_q)=c\in \comps(T_i)$, and let $q_1,q_2$
    be the children of $q$ with $V(G_{q_1})=c_1$ and $V(G_{q_2})=c_2$.
    Clearly $uv\notin E(G_{q_1})\cup E(G_{q_2})$, hence there must be
    an $\eta$\hy node $p$ somewhere on the path between $q$ and $r$
    where the edge $uv$ is introduced. However, since $v$ and $w$
    share the same label in $G_q$, and $u$ and $x$ share the same
    label in $G_q$, this also holds in~$G_p$. Hence the $\eta$\hy
    operation that introduces the edge $uv$ also introduces the edge
    $xw$. However, this contradicts the assumption that $xw\notin
    E$. Hence the path property holds as well. We conclude that~$\DDD$
    is indeed a $k$\hy derivation of~$G$.

    The above procedure for generating the $k$\hy derivation can clearly
    be carried out in polynomial time.
\end{proof}

\begin{example}\label{ex:expression->derivation}
  Consider the $3$\hy expression $\phi$ for the graph $P_4$ of
  Example~\ref{ex:p4}. Applying the procedure described in the proof
  of Lemma~\ref{lem:expression->derivation} we obtain the $3$\hy
  derivation $\DDD$ of Example~\ref{ex:derivation}.
\end{example}

\begin{lemma}\label{lem:derivation->expression}
  From a $k$\hy derivation of a graph $G$ we can obtain a $k$\hy
  expression of~$G$ in polynomial time.
\end{lemma}
\begin{proof}
  Let $\DDD=(T_0,\dots,T_t)$ be a $k$\hy derivation of $G=(V,E)$.
  Using the construction of the proof of Lemma~\ref{lem:make-strict}
  we can obtain a strict $k$\hy derivation of $G$ from any given
  $k$\hy derivation of $G$. Hence we may assume, w.l.o.g., that
  $\DDD$ is strict. Let $C=\bigcup_{i=0}^t \comps(T_i)$.
  We are going to construct in polynomial time a $k$\hy
  expression tree for $G$, which can clearly be turned into a $k$\hy
  expression for $G$ in polynomial time.

  We proceed in three steps.

  First we construct a $k$\hy expression tree $Q_\oplus$ that only
  contains $\oplus$\hy nodes and leaves. For each component $c=\{v\}$ of
  $T_0$ we introduce a leaf $q(c,0)$ with label $1(v)$.
  For each $1\leq i \leq t$ and
  each component $c\in \comps(T_i)$ we introduce an $\oplus$\hy node
  $q(c,i)$. We add edges to $Q_\oplus$ such that $q(c',i-1)$ is a
  child of $q(c,i)$ if and only if $c'\subseteq c$.  Properties D1
  and D3 of a derivation ensure that $Q_\oplus$ is a tree. Note that
  $Q_\oplus$ is not necessarily succinct, and may contain $\oplus$\hy
  nodes that have only one child.

  In the next step we add to $Q_\oplus$ certain $\rho$\hy nodes to
  obtain the $k$\hy expression tree $Q_{\oplus,\rho}$.  We visit the
  $\oplus$\hy nodes of $Q_\oplus$ in a depth-first ordering. Let
  $q(c,i)$ be the currently visited node. Between $q(c,i)$ and each
  child $q(c',i-1)$ of $q(c,i)$ we add at most $k$ $\rho$\hy nodes (so
  that the edge between $q(c,i)$ and $q(c',i-1)$ becomes a path) such
  that afterwards $q(c,i)$ has a child $q'$ with $\groups(G_{q'})=\SB
  g\in \groups(G_{q(c,i)}) \SM g\subseteq c \SE \subseteq
  \groups(T_i)$. This is possible because of properties D2 and D4 of a
  derivation.
 
  As a final step, we add $\eta$\hy nodes to $Q_{\oplus,\rho}$ and the
  $k$\hy expression tree~$Q$. Let $uv\in E$ be an edge of $G$. We show
  that there is an $\oplus$\hy node $q$ in $Q_{\oplus,\rho}$ above
  which we can add an $\eta$\hy node $p$ ($q$ is a child of $p$) which
  introduces edges including $uv$ but does not introduce any edge not
  present in~$E$ .

  Let $q(c,i)$ be the $\oplus$\hy node of $Q_{\oplus,\rho}$ with
  smallest $i$ such that $u,v \in V(G_{q(c,i)})$. We write $q=q(c,i)$
  and $c=V(G_q)$ and observe that $c\in \comps(T_i)$.  Among the
  children of $q$ are two distinct nodes $q_1,q_2$ such that $u\in
  V(G_{q_1})$ and $v\in V(G_{q_2})$.  It follows that there are
  distinct components $c_1,c_2\in T_{i-1}$ with $u\in c_1$ and $v\in
  c_2$. By the edge property, $u$ and $v$ belong to different groups
  of $T_i$, and so $u$ and $v$ have different labels in $G_q$, say the
  labels $a$ and $b$, respectively.  We add an $\eta$\hy node $p$
  above $q$ representing the operation $\eta_{a,b}$. This inserts the
  edge $uv$ to $G_q$. We need to show that $\eta_{a,b}$ does not add
  any edge that is not in~$E$.  We show that for all pairs of vertices
  $u',v' \in c$ where $u'$ has label $a$ and $v'$ has label $b$ in
  $G_q$, the edge $u'v'$ is in $E$.

  We consider four cases.

  Case 1: $u=u', v =v'$. Trivially,  $u'v'=uv\in E$.

  Case 2: $u=u', v\neq v'$.   Assume to the contrary that
  $u'v'\notin E$. Since $v$ and $v'$ have the same label in
  $G_q$, they belong to the same group of $T_i$. The neighborhood
  property implies that $u$ and $v$ belong to the same component of
  $T_{i-1}$, a contradiction to the minimal choice of~$i$. Hence
  $u'v'\in E$.

  Case 3: $u\neq u', v= v'$. This case is symmetric to Case~2.
 
  Case 4: $u\neq u', v\neq v'$.  Assume to the contrary that $u'v'
  \notin E$.  It follows by from Cases~2 and~3 that $uv',vu'\in
  E$. The path property implies that $u$ and $v$ belong to the same
  component of $T_{i-1}$, a contradiction to the minimal choice
  of~$i$. Hence $u'v'\in E$.

  Consequently, we can successively add $\eta$\hy nodes to
  $Q_{\oplus,\rho}$ until all edges of $E$ are inserted, but no edge
  outside of $E$. Hence we obtain indeed a $k$\hy expression tree for
  $G$.

  This procedure for generating the $k$\hy expression tree can clearly
  be carried out in polynomial time, hence the lemma follows.
\end{proof}

We note that we could have saved some $\rho$\hy operations in the
proof of Lemma~\ref{lem:derivation->expression}.  In particular the
$k$\hy expression produced may contain $\rho$\hy operations where the
number of different labels before and after the application of the
$\rho$\hy operation remains the same. It is easy to see that such a
$\rho$\hy operations can be omitted if we change labels of some
initial $k$\hy graphs accordingly.

\begin{example}
  Consider the derivation $\DDD$ of graph $G$ in
  Example~\ref{ex:derivation}.  We construct a $3$\hy expression of
  $G$ using the procedure as described in the proof of
  Lemma~\ref{lem:derivation->expression}, however, to save space, we
  give the construction in terms of $k$\hy expressions instead of
  $k$\hy expression trees.  First we obtain $\phi_\oplus= ((1(a)\oplus
  1(b)) \oplus 1(c)) \oplus 1(d)$.  Next we insert $\rho$ operations
  to represent how the groups evolve through the derivation:
  $\phi_{\oplus,\rho}= \rho_{1\rightarrow 2} ( ( 1(a)\oplus
  \rho_{1\rightarrow 2}(1(b) ) \oplus \rho_{1\rightarrow 3 } 1(c) ) )
  \oplus 1(d)$.  Finally we add $\eta$ operations, and obtain
  $\phi_{\oplus,\rho,\eta}= \eta_{1,2}( \rho_{1\rightarrow 2}(
  \eta_{2,3}( \eta_{1,2}( 1(a)\oplus \rho_{1\rightarrow 2}(1(b) ) )
  \oplus \rho_{1\rightarrow 3 } 1(c) ) ) \oplus 1(d) )$.
\end{example}

By Lemma~\ref{lem:strict-short} we do not need to search for $k$\hy
derivations of length $>n-1$ when the graph under consideration has
$n$ vertices.  The next lemma improves this bound to $n-k+1$ which
provides a significant improvement for our SAT encoding, especially if
the graph under consideration has large clique-width.

\begin{lemma}\label{lem:strict-shorter}
  Let $1\leq k\leq n$.  If a graph with $n$ vertices has a $k$\hy
  derivation, then it has a $k$\hy derivation of length $n-k+1$.
\end{lemma}
\begin{proof}
  Let $k\geq 1$ be fixed.  We define the \emph{$k$\hy length} of a
  derivation as the number of templates that contain at least one
  component of size larger than $k$ (these templates form a suffix of
  the derivation).  Let~$\ell(n,k)$ be the largest $k$\hy length of a
  strict derivation over a universe of size $n$.  Before we show the
  lemma, we establish three claims. For these claims, the groups of
  the considered derivations are irrelevant and hence we will be ignored.

  \smallskip\noindent\emph{Claim 1:}   $\ell(n,k)<\ell(n+1,k)$.
  
  To show the claim, consider a strict derivation
  $\DDD=(T_0,\dots,T_t)$ over a universe $V$ of size $n$ with $k$\hy
  length $\ell$. We take a new element $a$ and form a strict
  derivation $\DDD'$ over the universe $V \cup \{a\}$ by adding the
  singleton $\{a\}$ to $\comps(T_i)$ for $0\leq i \leq t$ and adding a
  new template $T_{t+1}$ with $\comps(T_{t+1})= \{V \cup \{a\}\}$. The
  new derivation $\DDD'$ has $k$\hy length $\ell+1$.

  \smallskip\noindent\emph{Claim 2:} Let $\DDD=(T_0,\dots,T_t)$ be a
  strict derivation over a universe $V$ of size $n$ of $k$\hy length
  $\ell(n,k)$. Then,  $T_{t-\ell(n,k)+1}$ has exactly one component of
  size $k+1$ and all other components are singletons.

  We proceed to show the claim. Let $j=t-\ell(n,k)$, and observe that
  $j$ is the largest index where all components of $T_j$ have size at
  most $k$. Let $c_1,\dots,c_r$ be the components of $T_{j+1}$ of size
  greater than $1$ such that $\Card{c_1} \geq \Card{c_2} \geq \dots
  \geq \Card{c_r}$. Thus $\Card{c_1}>k$.  We show that $r=1$. Assume
  to the contrary that $r>1$. We pick some element $a_i\in c_i$,
  $2\leq i \leq r$, and set $X= \bigcup_{i=2}^r c_i\setminus
  \{a_i\}$. The derivation $\DDD$ induces a strict derivation $\DDD'$
  over the universe $V'=V\setminus X$. Observe that
  $n'=\Card{V'}<\Card{V}=n$.  Evidently $\DDD'$ has the same $k$\hy
  length as $\DDD$, hence $\ell(n',k)\geq \ell(n,k)$, a contradiction
  to Claim~1. Hence $r=1$, and $c_1$ is the only component in
  $T_{j+1}$ of size greater than $k$, all other components of
  $T_{j+1}$ are singletons. We show that $\Card{c_1}=k+1$. We assume
  to the contrary that $\Card{c_1}>k+1$. We pick $k+1$ elements
  $b_1,\dots,b_{k+1} \in c_1$ and set $X=c_1\setminus
  \{b_1,\dots,b_{k+1}\}$.  Similarly as above, we observe that $\DDD$
  induces a strict derivation $\DDD''$ over the universe
  $V''=V'\setminus X$, and that $\DDD''$ has the same $k$\hy length as
  $\DDD$. Since $\Card{V''}<\Card{V}$ we have again a contradiction
  to~Claim~1.  Hence Claim~2 is established.

  \smallskip\noindent\emph{Claim 3:} $\ell(n,k)\leq n-k$. 

  To see the claim, let $\DDD=(T_0,\dots,T_t)$ be a strict derivation
  over a universe $V$ of size $n$ of $k$\hy length $\ell(n,k)$.  Let
  $j=t-\ell(n,k)$.  By Claim~2 we know that $T_{j+1}$ has exactly one
  component of size $k+1$ and all other components are singletons
  (hence there are $n-k-1$ singletons).  We conclude that
  $\Card{\comps(T_{j+1})}=n-k$. Since $\DDD$ is strict, we have 
  $n-k=\Card{\comps(T_{j+1})}>\Card{\comps(T_{j+2})}>\dots>
  \Card{\comps(T_{t})}=1$. Thus $\ell(n,k)=t-j\leq n-k$, and the claim
  follows.

  \medskip\noindent We are now in the position to establish the
  statement of the lemma.  Let $\DDD=(T_0,\dots,T_t)$ be a $k$\hy
  derivation of a graph $G=(V,E)$ with $\Card{V}=n$. By
  Lemma~\ref{lem:make-strict} we may assume that $\DDD$ is strict.
  Let~$\ell$ be the $k$\hy length of $\DDD$ and let $j=t-\ell$. By
  Claim~3 we know that $\ell\leq n-k$.  We define a new template
  $T_j'$ with $\comps(T_j')=\comps(T_j)$ and
  $\groups(T_j')=\groups(T_0)$, and we set
  $\DDD'=(T_0,T_j',T_{j+1},\dots,T_t)$. We claim that $\DDD'$ is a
  $k$\hy derivation of $G$. Clearly $\DDD'$ is a derivation, but we
  need to check the edge, neighborhood, and path property for $T_j'$
  and $T_{j+1}$ in $\DDD'$. The properties hold trivially for $T_j'$
  since all its groups are singletons. For $T_{j+1}$ the properties
  hold since $T_j'$ has the same components as $T_j$.  Thus $\DDD'$ is
  indeed a $k$\hy derivation of~$G$. The length of $\DDD'$ is
  $\ell+1\leq n-k+1$, hence the lemma follows.
\end{proof}

\begin{example}
  Again, consider the derivation $\DDD$ of Example~\ref{ex:derivation}.
  $\DDD$ defines $P_4$ which has clique-width
  3~\cite{CourcelleOlariu00}.  According to
  Lemma~\ref{lem:strict-shorter}, it should have a derivation of
  length $n - k + 1 = 4 - 3 + 1 = 2$.  We can obtain such a derivation
  by removing $T_1$ from $\DDD$, which gives $\DDD' = (T_0, T_2, T_3)$.
\end{example}

By combining Lemmas \ref{lem:expression->derivation},
\ref{lem:derivation->expression}, and \ref{lem:strict-shorter}, we
arrive at the main result of this section.
\begin{proposition}\label{pro:expression}
  Let $1\leq k \leq n$. A graph $G$ with $n$ nodes has clique-width at
  most $k$ if and only if $G$ has a $k$\hy derivation of length at
  most $n-k+1$.
\end{proposition}

 
\section{Encoding a Derivation of a Graph}
Let $G=(V,E)$ be graph, and $t>0$ an integer. We
are going to construct a CNF formula $\Fder(G,t)$ that is satisfiable if
and only if $G$ has a derivation of length~$t$.
We assume that the vertices of $G$ are given in some arbitrary but
fixed linear order~$<$.  

For any two distinct vertices $u$ and $v$ of $G$ and any $0\leq i \leq
t$ we introduce a {\em component variable} $c_{u,v,i}$.  Similarly, for any two
distinct vertices $u$ and $v$ of $G$ with $u<v$ and any $0\leq i \leq
t$ we introduce a {\em group variable} $g_{u,v,i}$. Intuitively, $c_{u,v,i}$ or
$g_{u,v,i}$ are true if and only if $u$ and $v$ are in the same component or group, respectively,
in the $i$th template of an implicitly represented derivation of $G$.

The formula $\Fder(G,t)$ is the conjunction of all the clauses described
below.

\noindent The following clauses represent the conditions~D1--D4.
\begin{myquote}
$(\bar c_{u,v,0}) \land (c_{u,v,t})
\land
(c_{u,v,i} \lor \bar g_{u,v,i}) 
\land
(\bar c_{u,v,i -1 } \lor c_{u,v,i})
\land
(\bar g_{u,v,i -1 } \lor g_{u,v,i})$ \\
 \phantom{x}\hfill for $u,v \in V$, $u<v$, $0\leq i \leq t$.
\end{myquote}

We further add clauses that ensure that the  relations of being in the
same group  and of being in the same component are transitive.
\begin{myquote}
  $(\bar c_{u,v,i} \lor \bar c_{v,w,i} \lor c_{u,w,i}) \land (\bar
  c_{u,v,i} \lor \bar c_{u,w,i} \lor c_{v,w,i}) \land (\bar c_{u,w,i}
  \lor \bar c_{v,w,i} \lor c_{u,v,i})\;\; \land $ 

$(\bar g_{u,v,i} \lor \bar g_{v,w,i} \lor g_{u,w,i}) \land (\bar g_{u,v,i} \lor \bar g_{u,w,i} \lor g_{v,w,i}) \land (\bar g_{u,w,i} \lor \bar g_{v,w,i} \lor g_{u,v,i})$ \\
\phantom{x}\hfill for $u,v,w\in V$, $u<v<w$, $0\leq i \leq t$.
\end{myquote}

\noindent In order to enforce the \emph{edge property} we add the
following clauses for any two vertices $u,v\in V$ with 
$u<v$, 
$uv\in E$ and
$1\leq i \leq t$:
\begin{myquote}
  $(c_{u,v,i-1} \lor \bar g_{u,v,i})$.
\end{myquote}
Further, to enforce the \emph{neighborhood property}, we add for any
three vertices $u,v,w\in V$ with $uv \in E$ and $uw \notin E$ and
$1\leq i \leq t$, the following clauses.
\begin{myquote}
  $(c_{\min(u,v),\max(u,v),i-1} \lor \bar g_{\min(v,w),\max(v,w),i})$ 
  


 \end{myquote}
Finally, to enforce the \emph{path property} we add for any four
vertices $u,v,w,x$, such that  $uv, uw, vx \in E$, and $wx \notin E$, $u < v$
and $1\leq i \leq t$ 
the following clauses:
\begin{myquote}
  ($c_{u,v,i-1} \lor \bar g_{\min(u,x),\max(u,x),i} \lor \bar g_{\min(v,w),\max(v,w),i})$ 



\end{myquote}

\medskip\noindent
The following statement is a direct consequence of the above
definitions.
\begin{lemma}\label{lem:der}
  $\Fder(G,t)$ is satisfiable if and only if $G$ has a derivation of
  length~$t$.
\end{lemma}

\section{Encoding a $k$\hy Derivation of a Graph}

In this section, we describe how the formula $\Fder(G,t)$ can be
extended to encode a derivation of width at most $k$.  Ideally, one
wants to encode that unit propagation results in a conflict on any
assignment of component and group variables representing a derivation
containing a component with more than $k$ groups.  First we will
describe the conventional direct encoding~\cite{Walsh00} followed by
our representative encoding. Only the latter encoding realizes arc consistency~\cite{Gent02}.

%

\subsection{Direct Encoding}
We introduce new Boolean variables $l_{v,a,i}$ for $v\in V$, $1\leq a \leq k$,
and $0\leq i \leq t$. The purpose is to assign each vertex for each
template a group number between~$1$ and~$k$. The intended meaning of a
variable $l_{v,a,i}$ is that in $T_i$, vertex $v$ has group number~$a$. 
Let $F(G,k,t)$ denote the formula obtained from $\Fder(G,t)$ by adding the
following three sets of clauses. The first ensures that every vertex has at
least one group number, the second ensures that every vertex has at
most one group number, and the third ensures that two vertices of the
same group share the same group number.
\begin{myquote}
  $(l_{v,1,i} \lor l_{v,2,i} \lor \dots \lor l_{v,k,i})$ \qquad for $v
  \in V$, $0\leq i \leq t$,

  $(\bar l_{v,a,i} \lor \bar l_{v,b,i})$ \qquad for $v \in V$, $1 \leq
  a < b \leq k$, $0\leq i \leq t$,

  $(\bar l_{u,a,i} \lor \bar l_{v,a,i} \lor \bar c_{u,v,i} \lor  g_{u,v,i}) \land (\bar
  l_{u,a,i} \lor l_{v,a,i} \lor \bar g_{v,w,i}) \land 
  (l_{v,a,i} \lor \bar l_{v,a,i} \lor  \bar g_{u,v,i})$\\[0.1cm]
  \phantom{x}\hfill for $u,v \in V$, $u < v$, $1 \leq a \leq k$, $0 \leq i \leq
t$.
\end{myquote}

\noindent Together with Lemma~\ref{lem:der} this construction directly yields
the following statement.
\begin{proposition}
  Let $G=(V,E)$ be graph and $t=\Card{V}-k +1$. Then  
  $F(G,k,t)$ is satisfiable if and only if $\cwd(G)\leq k$. 
\end{proposition}

\begin{example}\label{example:unitprop}
  Let $G=(V,E)$ and $k = 2$. Vertices $u, v, w \in
  V$ in template $T_i$,  are in one component, but in different groups.
  Hence the corresponding component variables are true,
  and the corresponding group variables are false.
  The clauses containing the variables $l_{u,a,i}, l_{v,a,i}, l_{w,a,i}$
  with $a \in \{1,2\}$ after removing falsified literals
  are: 
\begin{myquote}
$(l_{u,1,i} \lor l_{u,2,i}) \land (l_{v,1,i} \lor l_{v,2,i}) \land (l_{w,1,i} \lor l_{w,2,i}) \land 
(\bar l_{u,1,i} \lor \bar l_{v,1,i}) \land (\bar l_{u,1,i} \lor \bar l_{w,1,i}) \land \phantom{x}$ \\
$(\bar l_{v,1,i} \lor \bar l_{w,1,i}) \land
(\bar l_{u,2,i} \lor \bar l_{v,2,i}) \land (\bar l_{u,2,i} \lor \bar l_{w,2,i}) \land (\bar l_{v,2,i} \lor \bar l_{w,2,i})$
\end{myquote}
These clauses cannot be satisfied, yet unit propagation will not result in a conflict. 
Therefore, a SAT solver may not be able to cut off the current branch.
\end{example}

\subsection{The Representative Encoding}

To overcome the unit propagation problem of the direct encoding, as
described in Example~\ref{example:unitprop}, we propose the {\em
  representative encoding} which uses two types of variables. First,
for each $v\in V$ and $1\leq i \leq t$ we introduce a representative
variable $r_{v,i}$. This variable, if assigned to true, expresses that
vertex $v$ is the representative of a group in
template~$T_i$.  In each group, only one vertex can be the
representative and we choose to make the first vertex in the
lexicographical ordering the representative. 
This results in the following clauses:
\begin{myquote}
$( r_{v,i} \lor  \bigvee_{u\in V, u < v}  g_{u,v,i} ) \land \bigwedge_{u \in V, u < v} (\bar r_{v,i} \lor \bar g_{u,v,i} )$
\qquad for $v \in V$, $0\leq i \leq t$

\end{myquote}
\medskip\noindent Additionally we introduce auxiliary variables
to efficiently encode that the number of representative vertices in a
component is at most $k$.
These auxiliary variables are based on the {\em
  order encoding}~\cite{TamuraTagaKitagawaBanbara09}.  Consider a
(non-Boolean) variable $L_{v,i}$ with domain $D = \{1,\dots,k\}$,
whose elements denote the group number of vertex $v$ in template
$T_i$. In the direct encoding, we used $k$ variables $l_{v,a,i}$ with
$a \in D$. Assigning $l_{v,a,i} = {\tt 1}$ in that encoding means
$L_{v,i} = a$.  Alternatively, we can use {\em order variables}
$o^>_{v,a,i}$ with $v \in V$, $a \in D \setminus \{k\}$, $ 0 \leq i
\leq t$.  Assigning $o^>_{v,a,i} = {\tt 1}$ means $L_{v,i} >
a$. Consequently, $o^>_{v,a,i} = {\tt 0}$ means $L_{v,i} \leq a$.

\begin{examplenoqed} 
Given an assignment to the order variables $o^>_{v,a,i}$, one can easily construct the equivalent assignment to the variables in the 
direct encoding (and the other way around). Below is a visualization of the equivalence relation with $k=5$. In the middle is a binary 
representation of each of the $k$ labels by concatenating the Boolean values to the order variables.
\begin{myquote}
$\begin{array}[b]{c@{\;\;}c@{\;\;}c@{\;\;}c@{\;\;}c@{\;\;}c@{\;\;}l}
L_v = 1 &\leftrightarrow& l_{v,1,i} = {\tt 1} &\leftrightarrow&  {\tt 0000} &\leftrightarrow&  o^>_{v,1,i} = o^>_{v,2,i} = o^>_{v,3,i} = o^>_{v,4,i} = {\tt 0} \\
L_v = 2 &\leftrightarrow& l_{v,2,i} = {\tt 1} & \leftrightarrow  & {\tt 1000} & \leftrightarrow & o^>_{v,1,i} = {\tt 1}, o^>_{v,2,i} = o^>_{v,3,i} = o^>_{v,4,i} = {\tt 0} \\
L_v = 3 &\leftrightarrow& l_{v,3,i} = {\tt 1} & \leftrightarrow & {\tt 1100} &  \leftrightarrow & o^>_{v,1,i} =  o^>_{v,2,i} = {\tt 1}, o^>_{v,3,i} = o^>_{v,4,i} = {\tt 0} \\
L_v = 4 &\leftrightarrow& l_{v,4,i} = {\tt 1} & \leftrightarrow  & {\tt 1110} &  \leftrightarrow & o^>_{v,1,i} = o^>_{v,2,i} = o^>_{v,3,i} = {\tt 1}, o^>_{v,4,i} = {\tt 0} \\
L_v = 5 &\leftrightarrow& l_{v,5,i} = {\tt 1} & \leftrightarrow & {\tt 1111} &  \leftrightarrow & o^>_{v,1,i} = o^>_{v,2,i} = o^>_{v,3,i} = o^>_{v,4,i} = {\tt 1}
\end{array}$\exqed
\end{myquote}
\end{examplenoqed}

\medskip

Although our encoding is based on the variables from the order encoding, 
we use none of the associated clauses. We implemented the original 
order~\cite{TamuraTagaKitagawaBanbara09}, which resulted 
in many long clauses and the performance was comparable to the direct
encoding.

Instead, we combined the representative and order variables.  Our use
of the order variables can be seen as the encoding of a sequential
counter~\cite{Sinz05}.  We would like to point out that if $u$ and $v$
are both representative vertices in the same component of template
$T_i$ and $u < v$, then $o^>_{u,a,i} = {\tt 0}$ and $o^>_{v,a,i} =
{\tt 1}$ must hold for some $1\leq a < k$.  Consequently,
$o^>_{u,k-1,i} = {\tt 0}$ (vertex $u$ has not the highest group number
in $T_i$), \mbox{$o^>_{v,1,i} = {\tt 1}$} (vertex $v$ has not the
lowest group number in $T_i$), and \mbox{$o^>_{u,a,i} \rightarrow
  o^>_{v,a+1,i}$}:
These constraints can be expressed by the following clauses.
\begin{myquote}
  $(\bar c_{u,v,i} \lor \bar r_{u,i} \lor \bar r_{v,i} \lor \bar o^>_{u,k-1,i}) \land (\bar c_{u,v,i} \lor \bar r_{u,i} \lor \bar r_{v,i} \lor o^>_{v,1,i}) \land \\
  \bigwedge_{1 \leq a < k-1} (\bar c_{u,v,i} \lor \bar r_{u,i} \lor
  \bar r_{v,i} \lor \bar o^>_{u,a,i} \lor o^>_{v,a+1,i})$
\quad for $u,v \in V$, $u < v$, $0 \leq i \leq t$.
\end{myquote}

%

\begin{example}
  Consider a graph $G = (V,E)$ with $u,v,w,x \in V$ and the representative encoding with $k=3$.
  We will show that if $u$,$v$,$w$, and $x$ are all
  in the same component and they are all
  representatives of their respective group numbers in template $T_i$, 
  then unit propagation will result in a conflict (because there are four
  representatives and only three group numbers).  Observe that 
   all corresponding component and representative variables are true. 
  This example, with falsified literals removed, contains the clauses
    $(\bar o^>_{u,2,i}) $, 
    $ ( \bar o^>_{u,1,i} \lor \bm{o^>_{v,2,i}})$, 
    $ (o^>_{v,1,i}) $, $ (\bar o^>_{u,2,i}) $, 
    $ (\bar o^>_{u,1,i} \lor \bm{o^>_{w,2,i}}) $, 
    $ (o^>_{w,1,i}) $, 
    $(\bar o^>_{u,2,i}) $, 
    $ (\bar o^>_{u,1,i} \lor o^>_{x,2,i} ) $, 
    $(o^>_{x,1,i})$, 
    $(\bar o^>_{v,2,i})$, 
    \mbox{$ (\bm{\bar o^>_{v,1,i}} \lor \bm{o^>_{w,2,i}})$}, 
    $ (o^>_{w,1,i}) $, 
    $(\bar o^>_{v,2,i}) $, 
    $ (\bm{\bar o^>_{v,1,i}} \lor o^>_{x,2,i}) $, 
    $(o^>_{x,1,i}) $, 
    $(\bar o^>_{w,2,i}) $, 
    $ (\bm{\bar o^>_{w,1,i}} \lor o^>_{x,2,i}) $, 
    $ (o^>_{x,1,i})$.
Literals that are falsified by unit clauses are shown in bold. 
Notice that $({\bar o^>_{v,1,i} \lor {o^>_{w,2,i}}})$ is falsified, i.e., a conflicting clause.
\end{example}

Both the direct and representative encoding require $n(n + k - 1)(n -
k + 2)$ variables.  The number of clauses depends on the set of edges. In
worst case, the number of clauses can be $\mathcal{O}(n^5 - n^4k)$ due
to the path condition.  

\section{Experimental Results}

In this section we report the results we obtained by running our SAT
encoding on various classes of graphs.  Given a graph $G= (V,E)$, we
compute that $G$ has clique-width $k$ by determining for which value
of $k$ it holds that $F(G,k, |V| - k + 1 )$ is satisfiable and $F(G,k-1,
|V| - k + 2)$ is unsatisfiable.
We used the SAT solver {\tt Glucose} version 2.2~\cite{Glucose}  to solve the encoded problems. 
{\tt Glucose} solved the hardest instances about twice as fast (or more) as 
other state-of-the-art solvers such as {\tt Lingeling}~\cite{Biere12}, {\tt Minisat}~\cite{Minisat} and {\tt Clasp}~\cite{Clasp}.
We used a 4-core Intel Xeon CPU E31280 3.50GHz,
32 Gb RAM machine running Ubuntu 10.04.

Although the direct and representative encodings result
in CNF formulas of almost equal size, there is a huge difference in
 costs to solve these instances. To determine the
clique-width of the famous named graphs (see below)
using the direct encoding takes about two to three orders of magnitude
longer as compared to the representative encoding.  For example, we can
establish that the Paley graph with 13 vertices has clique-width 9
within a few seconds using the representative encoding, while the
solver requires over an hour using the direct encoding. Because of the
huge difference in speed, we discarded the use of the direct encoding
in the remainder of this section.

We noticed that upper bounds (satisfiable formulas) are obtained
much faster than lower bounds (unsatisfiable formulas). The
reason is twofold.  First,  the whole search space
needs to be explored for lower bounds, while for upper bounds, one can be ``lucky"
and find a solution fast. Second, due to our encoding, upper bound
formulas are smaller (due to a smaller $t$) which makes them
easier. Table~\ref{tab:random20} shows this for a random graph with 20
vertices for the direct encoding and the representative
encoding. 

\begin{table}[htb]
\centering
\caption{Runtimes in seconds of the direct and representative encoding on a random graph
  with 20 vertices and 95 edges for different values of $k$. Up to $k=9$ the formulas
  are unsatisfiable, afterwards they are satisfiable. Timeout (TO) is 20,000 seconds.}
\label{tab:random20}

\medskip\small\setlength{\tabcolsep}{4pt}
\begin{tabular}{@{}l@{~~}|@{~~}ccccccc@{~~}|@{~~}ccccccc@{}}
\toprule
~~~$k$ & 3 & 4 & 5 & 6 & 7 & 8 & 9 & 10 & 11 & 12 & 13 & 14 & 15 & 16\\
\midrule
direct & 1.39 & 14.25 & 101.1 & 638.5 & 18,337 & TO & TO & TO & TO & 30.57 & 0.67 & 0.50 & 0.10 & 0.10\\
repres &  0.62 & 2.12 & 8.14 & 12.14 & 33.94 & 102.3 & 358.6 &9.21 & 0.40 & 0.35 &0.32 & 0.29 & 0.29 & 0.28\\
\toprule
\end{tabular}
\end{table}

We examined whether adding symmetry-breaking predicates could improve
performance. We used {\tt Saucy} version 3 for this
purpose~\cite{KatebiSakallahMarkov12}. After the addition of the
clauses with representative variables, the number of symmetries is drastically reduced.
However, one can generate symmetry-breaking
predicates for $\Fder(G,t)$ and add those instead. Although it is helpful
in some cases, the average speed-up was between 5 to 10\%.

Our experimental computations are ongoing. Below we report on some
of the results we have obtained so far.

\subsection{Random Graphs}

The asymptotics of the clique-width of random graphs have been studied
by Lee et al.~\cite{LeeLeeOum12}.  Their results show that for random
graphs on $n$ vertices the following holds asymptotically almost surely: If the
graphs are very sparse, with an edge probability below $1/n$, then
clique-width is at most 5; if the edge probability is larger than
$1/n$, then the clique-width grows at least linearly in $n$.  Our first group
of experiments complements these asymptotic results and provides a
detailed picture on the clique-width of small random graphs.  We have
used the SAT encoding to compute the clique-width of graphs with 10,
15, and 20 vertices, with the edge probability ranging from 0 to~1. A
plot of the distribution is displayed in Figure~\ref{fig:random}.  It is
interesting to observe the symmetry at edge probability $1/2$, and the
how the steepness of the curve increases with the number of
vertices. Note the ``shoulders'' of the curve for very sparse and very
dense graphs.
\begin{figure}[htb]
\centering
\includegraphics[width=11cm]{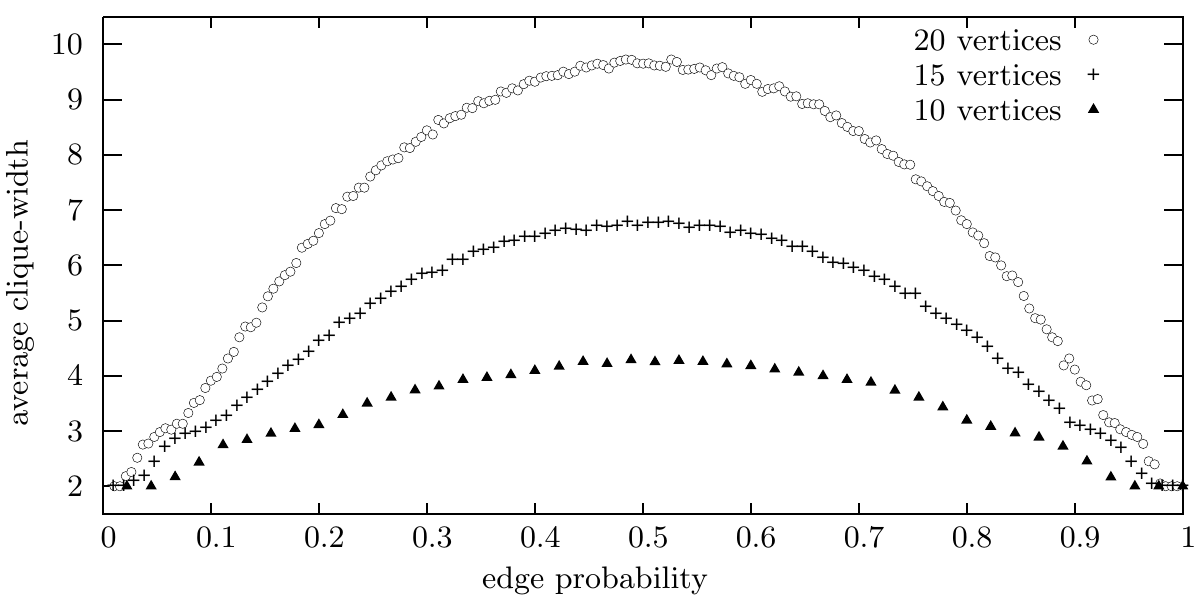}
\caption{Average clique-width of random graphs with edge probabilities between 0
  and~1. Each dot in the graph represents the average
  clique-width of 100 graphs. 
\label{fig:random}}
\end{figure}
  
\subsection{The Clique-Width Numbers}
 
For every $k>0$, let $n_k$ denote the smallest number such that there
exists a graph with $n_k$ vertices that has clique-width~$k$. We call
$n_k$ the $k$th \emph{clique-width number}.  From the
characterizations known for graphs of clique-width 1, 2, and 3,
respectively~\cite{CornelEtal12}, it is easy to determine the first
three clique-width numbers ($1$, $2$, and $4$). However, determining
$n_4$ is not straightforward, as it requires nontrivial arguments to
establish clique-width lower bounds.  We would like to point out that
a similar sequence for the graph invariant \emph{treewidth} is easy to
determine, as the complete graph on $n$ vertices is the smallest graph
of treewidth $n-1$.  One of the very few known graph classes of
unbounded clique-width for which the exact clique-width can be
determined in polynomial time are
grids~\cite{HeggernesMeisterRotics11}; the $k \times k$ grid with $k
\geq 3$ has clique-width $k+1$~\cite{GolumbicRotics00}.  Hence grids
provide the upper bounds $n_4 \leq 9$, $n_5\leq 16$, $n_6 \leq 25$,
and $n_7 \leq 36$.  With our experiments we could determine $n_4=6$,
$n_5=8$, $n_6 = 10$, $n_7=11$, $n_8 \leq 12$, and $n_9 \leq 13$.  It
is known that the path on four vertices ($P_4$) is the unique smallest
graph in terms of the number of vertices with clique-width 3. We could
determine that the triangular prism (\mbox{3-Prism}) is the unique
smallest graph with clique-width $4$, and that there are exactly 7
smallest graphs with clique-width $5$.
There are 68 smallest graphs with clique-width 6 and one of them has
only 18 edges. See Figure~\ref{fig:smallest} for an illustration. 
Additionally, we found several graphs of size 11 with clique-width 7
by extending a graph of size 10 with clique-width 6.

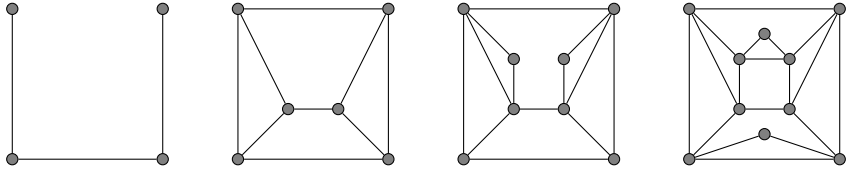
\begin{figure}[ht]
  \centering
  \tikzstyle{every node}=[circle,draw,inner sep=1.5pt,fill=gray]
    \begin{tikzpicture}
      \draw 
      (0,0)      node (a) {}
      (0,2)      node (b) {}
      (2,0)      node (c) {}
      (2,2)      node (d) {}
      (b)--(a)--(c)--(d)
      ;
    \begin{scope}[xshift=3cm]
      \draw 
      (0,0)      node (a) {}
      (0,2)      node (b) {}
      (2,0)      node (c) {}
      (2,2)      node (d) {}
      (0.667,0.667)    node (e) {}
      (1.333,0.667)    node (g) {}

      (a)--(b)--(d)--(c)--(a)
      (a)--(e)--(b)
      (c)--(g)--(d)
      (e)--(g)
      ;
    \end{scope}
    
        \begin{scope}[xshift=6cm]
      \draw 
      (0,0)      node (a) {}
      (0,2)      node (b) {}
      (2,0)      node (c) {}
      (2,2)      node (d) {}
      (0.667,0.667)    node (e) {}
      (0.667,1.333)    node (f) {}
      (1.333,0.667)    node (g) {}
      (1.333,1.333)    node (h) {}
      (a)--(b)--(d)--(c)--(a)
      (f)--(e)--(g)--(h)
     (a)--(e)
     (b)--(e)
     (b)--(f)
     (c)--(g)
     (d)--(h)
     (d)--(g)
      ;
    \end{scope}
    
        \begin{scope}[xshift=9cm]
      \draw 
      (0,0)      node (a) {}
      (0,2)      node (b) {}
      (2,0)      node (c) {}
      (2,2)      node (d) {}
      (0.667,0.667)    node (e) {}
      (0.667,1.333)    node (f) {}
      (1.333,0.667)    node (g) {}
      (1.333,1.333)    node (h) {}
      (1,0.333)    node (i) {}
      (1,1.667)    node (j) {}
      (a)--(b)--(d)--(c)--(a)
      (e)--(f)--(h)--(g)--(e)
     (a)--(e)
     (b)--(e)
     (b)--(f)
     (c)--(g)
     (d)--(h)
     (d)--(g)
      (a)--(i)--(c)
      (f)--(j)--(h)
      ;
    \end{scope}
\end{tikzpicture}
\caption{Smallest graphs with clique-width 3, 4, 5, and 6 (from left to right).}
  \label{fig:smallest}

\end{figure}

\begin{proposition}
  The clique-width sequence starts with the numbers $1$, $2$, $4$,
  $6$, $8$, $10$, $11$.
\end{proposition}
We used Brendan McKay's software package Nauty \cite{Mackay81} to
avoid checking isomorphic copies of the same graph.  There are several
other preprocessing methods that can speed up the search for small
graphs of prescribed clique-width $k\geq 2$.  Obviously, we can limit
the search to \emph{connected} graphs, as the clique-width of a graph
is clearly the maximum clique-width of its connected components. We
can also ignore graphs that contain \emph{twins}---two vertices
that have exactly the same neighbors---as we can delete one of them
without changing the clique-width. Similarly, we can ignore graphs
with a \emph{universal vertex}, a vertex that is adjacent to all other
vertices, as it can be deleted without changing the clique-width. All
these filtering steps are subsumed by the general concept of
\emph{prime graphs}.  Consider a graph $G=(V,E)$. A vertex $u\in V$
\emph{distinguishes} vertices $v,w\in V$ if $uv\in E$ and $uw\notin
E$. A set $M\subseteq V$ is a \emph{module} if no vertex from
$V\setminus M$ distinguishes two vertices from $M$. A module $M$ is
\emph{trivial} if $\Card{M}\in \{0,1,\Card{V}\}$. A graph is
\emph{prime} if it contains only trivial modules. It is well-known
that the clique-width of a graph is either $2$ or the maximum
clique-width of its induced prime
subgraphs~\cite{CourcelleOlariu00}. Hence, in our search, we can
ignore all graphs that are not prime. We can efficiently check whether a graph 
is prime~\cite{HabibPaul10}. The larger the number of
vertices, the larger the fraction of non-prime graphs (considering
connected graphs modulo isomorphism).  Table~\ref{table:sequence}
gives detailed results.

\begin{table}[htb]
\caption{Number of connected and prime graphs with specified clique-width, modulo isomorphism.}\label{table:sequence}
\centering
\medskip
\begin{tabular}{@{}c@{~~~~~~~~}r@{~~~}r@{~~~~~~}@{~~~}r@{~~~}r@{~~~}r@{~~~}r@{~~~}r@{}}
\toprule
&&& \multicolumn{5}{c}{clique-width}\\
\cmidrule{4-8}
$|V|$ & connected & prime~ & $2$ & $3$ & $4$ & $5$ & ~~$6$ \\ 
\midrule
4 &             6 &        1   &  0   &          1   &              0   &       0  & 0\\
5 &           21 &       4   &  0   &          4   &              0   &       0  & 0\\
6 &        112 &      26  & 0   &        25   &              1  &        0 & 0\\
7 &        853 &    260  &  0   &      210  &            50   &       0 & 0\\
8 &   11,117 & 4,670  &  0  &     1,873  &      2,790  &        7 & 0\\
9 &   261,080  & 145,870  & 0&  16,348 &   125,364  & 4,158 & 0\\
10 &   11,716,571  & 8,110,354  & 0&  142,745 &   5,520,350  & 2,447,190 & 68\\

\toprule
\end{tabular}
\end{table}

\subsection{Famous Named Graphs}
\label{sec:famous}

The graph theoretic literature contains several graphs that have
names, sometimes inspired by the graph's topology, and sometimes after
their discoverer. 
We have computed the clique-width of several named graphs,
the results are given in Table~\ref{table:named} (definitions of all
considered graphs can be found in MathWorld \cite{MathWorld}).  The
\emph{Paley graphs}, named after the English mathematician Raymond
Paley (1907--1933), stick out as having large clique-width. Our
results on the clique-width of Paley graphs imply some upper bounds on
the 9th and 11th clique-width numbers: $n_9\leq 13$ and $n_{11}\leq
17$.
 
\begin{table}[htb]
\centering
\caption{Clique-width of named graphs. Sizes are reported for the unsatisfiables.}\label{table:named}
\medskip
\begin{tabular}{l@{\hspace{5mm}}r@{\hspace{5mm}}r@{\hspace{5mm}}r@{\hspace{5mm}}r@{\hspace{5mm}}r@{\hspace{5mm}}r@{\hspace{5mm}}r}
\toprule
graph & $|V|$ & $|E|$& cwd & variables & clauses & UNSAT  & SAT \\ 
\midrule
Brinkmann & 21 & 42 & 10 & 8,526 & 163,065 & 3,932.56 & 1.79\\
Chv\'{a}tal   &                12 & 24 & 5 & 1,800 & 21,510 & 0.40 & 0.09 \\
Clebsch &                 16 & 40 & 8 & 3,872 & 60,520 & 191.02 & 0.09\\
Desargues &             20 & 30 & 8 & 7,800 & 141,410 & 3,163.70 & 0.26\\
Dodecahedron &      20 & 30 & 8 & 7,800 & 141,410 & 5,310.07 & 0.33\\
Errera     &                17 & 45 & 8 & 4,692  & 79,311 & 82.17 & 0.16\\
Flower snark &      20 & 30 & 7 & 8,000 & 148,620 & 276.24 & 3.9\\
Folkman &                20 & 40 & 5 & 8,280 & 168,190 & 11.67 & 0.36\\
Franklin &               12 & 18 & 4 & 1,848 & 21,798 & 0.07 & 0.04\\
Frucht   &               12 & 18 & 5 & 1,800 & 20,223 & 0.39 & 0.02\\
Hoffman &           16  & 32 & 6 & 4,160 & 64,968 & 8.95 & 0.46\\
Kittell  &                23  & 63 & 8 & 12,006 & 281,310 & 179.62 & 18.65\\
McGee &              24  & 36 & 8 & 13,680 &  303,660 & 8,700.94 & 59.89\\
Sousselier &       16  & 27 & 6 & 4,160  & 63,564 & 3.67 & 11.75\\
Paley-13  &         13 & 39 & 9 & 1,820 & 22,776 & 12.73 & 0.05 \\ 
Paley-17 &         17 & 68 & 11 & 3,978 & 72,896 & 194.38 & 0.12\\ 
Pappus  &            18 & 27 & 8 & 5,616 & 90,315 & 983.67 & 0.14\\
Petersen   &         10 & 15 & 5 & 1,040 & 9,550 & 0.10 & 0.02\\ 
Poussin &            15 & 39 & 7 & 3,300 & 50,145 & 9.00 & 0.21\\
Robertson &        19 & 38 & 9 & 6,422 & 112, 461 & 478.83 & 0.76\\ 
Shrikhande &       16 & 48 & 9 & 3,680 & 59,688 & 129.75 & 0.11\\
\toprule
\end{tabular}
\end{table}

\section{Conclusion}

We have presented a SAT approach to the exact computation of
clique-width, based on a reformulation of clique-width and several
techniques to speed up the search.  This new approach allowed us to
systematically compute the exact clique-width of various small graphs.  
We think that our results could be of relevance for theoretical investigations. 
For instance, knowing small vertex-minimal graphs of certain clique-width 
could be helpful for the design
of discrete algorithms that recognize graphs of bounded
clique-width. Such graphs can also be useful as gadgets for a
reduction to show that the recognition of graphs of clique-width~4 is
$\NP$\hy hard, which is still a long-standing open problem
\cite{FellowsRosamondRoticsSzeider09}.  Furthermore, as discussed in
Section~\ref{sect:intro}, there are no heuristic algorithms to compute
the clique-width directly, but heuristic algorithms for related
parameters can be used to obtain upper bounds on the clique-width. Our
SAT-based approach can be used to empirically evaluate how far
heuristics are from the optimum, at least for small and
medium-sized graphs. 

So far we have focused in our experiments on the exact clique-width,
but for various applications it is sufficient to have  good upper
bounds. Our results (see Table~\ref{tab:random20}) suggest that our
approach can be scaled to medium-sized graphs for the computation of
upper bounds.  We also observed that for many graphs the upper bound
of Lemma~\ref{lem:strict-shorter} is not tight. Thus, we expect that
if we search for shorter derivations, which is significantly faster,
this will yield optimal or close to optimal solutions in many cases.

Finally, we would like to mention that our SAT-based approach is very
flexible and open. It can easily be adapted to variants of
clique-width, such as linear clique-width
\cite{HeggernesMeisterPapadopoulos12,FellowsRosamondRoticsSzeider09},
$m$\hy clique-width \cite{CourcelleTwigg10}, or
NLC-width~\cite{Wanke94}. Hence, our approach can be used for an
empirical comparison of these parameters.

\section*{Acknowledgement}
The authors acknowledge the Texas Advanced Computing Center (TACC) at The University of Texas at 
Austin for providing grid resources that have contributed to the research results reported within this paper.


\begin{thebibliography}{10}

\bibitem{Glucose}
Gilles Audemard and Laurent Simon.
\newblock Predicting learnt clauses quality in modern sat solvers.
\newblock In {\em Proceedings of the 21st international jont conference on
  Artifical intelligence}, IJCAI'09, pages 399--404, San Francisco, CA, USA,
  2009. Morgan Kaufmann Publishers Inc.

\bibitem{Beyss2013}
Martin Bey\ss.
\newblock Fast algorithm for rank-width.
\newblock In {\em Mathematical and Engineering Methods in Computer Science, 8th
  International Doctoral Workshop, MEMICS 2012, Znojmo, Czech Republic, October
  25-28, 2012, Revised Selected Papers}, volume 7721 of {\em Lecture Notes in
  Computer Science}, pages 82--93. Springer Verlag, 2013.

\bibitem{Biere12}
Armin Biere.
\newblock Lingeling and friends entering the {SAT} {Challenge} 2012.
\newblock In A.~Balint, A.~Belov, A.~Diepold, S.~Gerber, M.~J\"{a}rvisalo, and
  C.~Sinz, editors, {\em Solver and Benchmark Descriptions}, volume B-2012-2 of
  {\em Department of Computer Science Series of Publications B.}, pages 33--34.
  University of Helsinki, 2012.

\bibitem{BuixuanTelleVatshelle11}
Binh-Minh Bui-Xuan, Jan~Arne Telle, and Martin Vatshelle.
\newblock Boolean-width of graphs.
\newblock {\em Theoretical Computer Science}, 412(39):5187--5204, 2011.

\bibitem{CornelEtal12}
Derek~G. Corneil, Michel Habib, Jean-Marc Lanlignel, Bruce Reed, and Udi
  Rotics.
\newblock Polynomial-time recognition of clique-width {$\leq 3$} graphs.
\newblock {\em Discr. Appl. Math.}, 160(6):834--865, 2012.

\bibitem{CornelRotics05}
Derek~G. Corneil and Udi Rotics.
\newblock On the relationship between clique-width and treewidth.
\newblock {\em SIAM J. Comput.}, 34(4):825--847, 2005.

\bibitem{CourcelleMakowskyRotics00}
B.~Courcelle, J.~A. Makowsky, and U.~Rotics.
\newblock Linear time solvable optimization problems on graphs of bounded
  clique-width.
\newblock {\em Theory Comput. Syst.}, 33(2):125--150, 2000.

\bibitem{CourcelleMakowskyRotics01}
B.~Courcelle, J.~A. Makowsky, and U.~Rotics.
\newblock On the fixed parameter complexity of graph enumeration problems
  definable in monadic second-order logic.
\newblock {\em Discr. Appl. Math.}, 108(1-2):23--52, 2001.

\bibitem{CourcelleOlariu00}
B.~Courcelle and S.~Olariu.
\newblock Upper bounds to the clique-width of graphs.
\newblock {\em Discr. Appl. Math.}, 101(1-3):77--114, 2000.

\bibitem{CourcelleEngelfrietRozenberg90}
Bruno Courcelle, Joost Engelfriet, and Grzegorz Rozenberg.
\newblock Context-free handle-rewriting hypergraph grammars.
\newblock In Hartmut Ehrig, Hans-J{\"o}rg Kreowski, and Grzegorz Rozenberg,
  editors, {\em Graph-Grammars and their Application to Computer Science, 4th
  International Workshop, Bremen, Germany, March 5--9, 1990, Proceedings},
  volume 532 of {\em Lecture Notes in Computer Science}, pages 253--268, 1991.

\bibitem{CourcelleEngelfrietRozenberg93}
Bruno Courcelle, Joost Engelfriet, and Grzegorz Rozenberg.
\newblock Handle-rewriting hypergraph grammars.
\newblock {\em J. of Computer and System Sciences}, 46(2):218--270, 1993.

\bibitem{CourcelleTwigg10}
Bruno Courcelle and Andrew Twigg.
\newblock Constrained-path labellings on graphs of bounded clique-width.
\newblock {\em Theory Comput. Syst.}, 47(2):531--567, 2010.

\bibitem{Diestel00}
Reinhard Diestel.
\newblock {\em Graph Theory}, volume 173 of {\em Graduate Texts in
  Mathematics}.
\newblock Springer Verlag, New York, 2nd edition, 2000.

\bibitem{DowKorf07}
P.~Alex Dow and Richard~E. Korf.
\newblock Best-first search for treewidth.
\newblock In {\em Proceedings of the Twenty-Second AAAI Conference on
  Artificial Intelligence, July 22-26, 2007, Vancouver, British Columbia,
  Canada}, pages 1146--1151. AAAI Press, 2007.

\bibitem{Minisat}
Niklas E\'en and Niklas S\"orensson.
\newblock An extensible sat-solver.
\newblock In Enrico Giunchiglia and Armando Tacchella, editors, {\em Theory and
  Applications of Satisfiability Testing}, volume 2919 of {\em Lecture Notes in
  Computer Science}, pages 502--518. Springer Berlin Heidelberg, 2004.

\bibitem{FellowsRosamondRoticsSzeider09}
Michael~R. Fellows, Frances~A. Rosamond, Udi Rotics, and Stefan Szeider.
\newblock Clique-width is {NP}-complete.
\newblock {\em SIAM J. Discrete Math.}, 23(2):909--939, 2009.

\bibitem{Clasp}
M.~Gebser, B.~Kaufmann, A.~Neumann, and T.~Schaub.
\newblock clasp: A conflict-driven answer set solver.
\newblock In C.~Baral, G.~Brewka, and J.~Schlipf, editors, {\em Proceedings of
  the Ninth International Conference on Logic Programming and Nonmonotonic
  Reasoning (LPNMR'07)}, volume 4483 of {\em Lecture Notes in Artificial
  Intelligence}, pages 260--265. Springer-Verlag, 2007.

\bibitem{Gent02}
Ian~P. Gent.
\newblock Arc consistency in {SAT}.
\newblock In F.~{van Harmelen}, editor, {\em 15th European Conference on
  Artificial Intelligence (ECAI 2002)}, pages 121--125. IOS Press, 2002.

\bibitem{GogateDechter04}
Vibhav Gogate and Rina Dechter.
\newblock A complete anytime algorithm for treewidth.
\newblock In {\em Proceedings of the Proceedings of the Twentieth Conference
  Annual Conference on Uncertainty in Artificial Intelligence (UAI-04)}, pages
  201--208, Arlington, Virginia, 2004. AUAI Press.

\bibitem{GolumbicRotics00}
Martin~Charles Golumbic and Udi Rotics.
\newblock On the clique-width of some perfect graph classes.
\newblock {\em Internat. J. Found. Comput. Sci.}, 11(3):423--443, 2000.
\newblock Selected papers from the Workshop on Graph-Theoretical Aspects of
  Computer Science (WG 99), Part 1 (Ascona).

\bibitem{HabibPaul10}
Michel Habib and Christophe Paul.
\newblock A survey of the algorithmic aspects of modular decomposition.
\newblock {\em Computer Science Review}, 4(1):41--59, 2010.

\bibitem{HeggernesMeisterPapadopoulos12}
Pinar Heggernes, Daniel Meister, and Charis Papadopoulos.
\newblock Characterising the linear clique-width of a class of graphs by
  forbidden induced subgraphs.
\newblock {\em Discr. Appl. Math.}, 160(6):888--901, 2012.

\bibitem{HeggernesMeisterRotics11}
Pinar Heggernes, Daniel Meister, and Udi Rotics.
\newblock Computing the clique-width of large path powers in linear time via a
  new characterisation of clique-width.
\newblock In Alexander~S. Kulikov and Nikolay~K. Vereshchagin, editors, {\em
  Computer Science - Theory and Applications - 6th International Computer
  Science Symposium in Russia, CSR 2011, St. Petersburg, Russia, June 14-18,
  2011. Proceedings}, volume 6651 of {\em Lecture Notes in Computer Science},
  pages 233--246. Springer Verlag, 2011.

\bibitem{HvidevoldEtal11}
Eivind~Magnus Hvidevold, Sadia Sharmin, Jan~Arne Telle, and Martin Vatshelle.
\newblock Finding good decompositions for dynamic programming on dense graphs.
\newblock In D{\'a}niel Marx and Peter Rossmanith, editors, {\em Parameterized
  and Exact Computation - 6th International Symposium, IPEC 2011,
  Saarbr{\"u}cken, Germany, September 6-8, 2011. Revised Selected Papers},
  volume 7112 of {\em Lecture Notes in Computer Science}, pages 219--231.
  Springer Verlag, 2012.

\bibitem{KatebiSakallahMarkov12}
Hadi Katebi, Karem~A. Sakallah, and Igor~L. Markov.
\newblock Conflict anticipation in the search for graph automorphisms.
\newblock In Nikolaj Bj{\o}rner and Andrei Voronkov, editors, {\em Logic for
  Programming, Artificial Intelligence, and Reasoning - 18th International
  Conference, LPAR-18, M{\'e}rida, Venezuela, March 11-15, 2012. Proceedings},
  volume 7180 of {\em Lecture Notes in Computer Science}, pages 243--257.
  Springer Verlag, 2012.

\bibitem{KosterBodlaenderHoesel01}
Arie M. C.~A. Koster, Hans~L. Bodlaender, and Stan P.~M. van Hoesel.
\newblock Treewidth: Computational experiments.
\newblock {\em Electronic Notes in Discrete Mathematics}, 8:54--57, 2001.

\bibitem{LeeLeeOum12}
Choongbum Lee, Joonkyung Lee, and Sang-il Oum.
\newblock Rank-width of random graphs.
\newblock {\em J. Graph Theory}, 70(3):339--347, 2012.

\bibitem{Mackay81}
Brendan~D. McKay.
\newblock Practical graph isomorphism.
\newblock In {\em Proceedings of the {T}enth {M}anitoba {C}onference on
  {N}umerical {M}athematics and {C}omputing, {V}ol. {I} ({W}innipeg, {M}an.,
  1980)}, volume~30, pages 45--87, 1981.

\bibitem{Oum08}
{Sang-il} Oum.
\newblock Approximating rank-width and clique-width quickly.
\newblock {\em ACM Transactions on Algorithms}, 5(1), 2008.

\bibitem{OumSeymour06}
{Sang-il} Oum and P.~Seymour.
\newblock Approximating clique-width and branch-width.
\newblock {\em J. Combin. Theory Ser. B}, 96(4):514--528, 2006.

\bibitem{SamerVeith09}
Marko Samer and Helmut Veith.
\newblock Encoding treewidth into {SAT}.
\newblock In {\em Theory and Applications of Satisfiability Testing - SAT 2009,
  12th International Conference, SAT 2009, Swansea, UK, June 30 - July 3, 2009.
  Proceedings}, volume 5584 of {\em Lecture Notes in Computer Science}, pages
  45--50. Springer Verlag, 2009.

\bibitem{Sinz05}
Carsten Sinz.
\newblock Towards an optimal cnf encoding of boolean cardinality constraints.
\newblock In Peter van Beek, editor, {\em Principles and Practice of Constraint
  Programming - CP 2005, 11th International Conference, CP 2005, Sitges, Spain,
  October 1-5, 2005, Proceedings}, volume 3709 of {\em Lecture Notes in
  Computer Science}, pages 827--831. Springer Verlag, 2005.

\bibitem{SmithUlusalHicks12}
J.~Cole Smith, Elif Ulusal, and Illya~V. Hicks.
\newblock A combinatorial optimization algorithm for solving the branchwidth
  problem.
\newblock {\em Comput. Optim. Appl.}, 51(3):1211--1229, 2012.

\bibitem{TamuraTagaKitagawaBanbara09}
Naoyuki Tamura, Akiko Taga, Satoshi Kitagawa, and Mutsunori Banbara.
\newblock Compiling finite linear {CSP} into {SAT}.
\newblock {\em Constraints}, 14(2):254--272, 2009.

\bibitem{Walsh00}
Toby Walsh.
\newblock {SAT} v {CSP}.
\newblock In R.~Dechter, editor, {\em 6th International Conferenc on Principles
  and Practice of Constraint Programming (CP 2000)}, volume 1894 of {\em
  Lecture Notes in Computer Science}, pages 441--456. Springer Verlag, 2000.

\bibitem{Wanke94}
Egon Wanke.
\newblock {$k$}-{NLC} graphs and polynomial algorithms.
\newblock {\em Discr. Appl. Math.}, 54(2-3):251--266, 1994.
\newblock Efficient algorithms and partial $k$-trees.

\bibitem{MathWorld}
Eric Weisstein.
\newblock {MathWorld} online maathematics resource.

\end{thebibliography}
 
\end{document}